\RequirePackage{amsthm}
\documentclass[sn-mathphys,Numbered,pdflatex]{sn-jnl}% Math and Physical Sciences Reference Style
\usepackage{graphicx}%
\usepackage{multirow}%
\usepackage{amsmath,amssymb,amsfonts}%
\usepackage{amsthm}%
\usepackage{mathrsfs}%
\usepackage[title]{appendix}%
\usepackage{xcolor}%
\usepackage{textcomp}%
\usepackage{manyfoot}%
\usepackage{booktabs}%
\usepackage{algorithm}%
\usepackage{algorithmicx}%
\usepackage{algpseudocode}%
\usepackage{listings}%

\usepackage{diagbox}

\usepackage{overpic}
\usepackage{bbm}

\theoremstyle{thmstyleone}%
\newtheorem{theorem}{Theorem}%  meant for continuous numbers
\newtheorem{proposition}[theorem]{Proposition}% 
\newtheorem{corollary}[theorem]{Corollary}% 
\newtheorem{lemma}[theorem]{Lemma}% 

\theoremstyle{thmstyletwo}%
\newtheorem{remark}{Remark}%

\theoremstyle{thmstylethree}%
\newtheorem{definition}{Definition}%
\newtheorem{construction}{Construction}

\usepackage{mathtools}
\DeclarePairedDelimiter\abs{\lvert}{\rvert}

\DeclarePairedDelimiter\ceil{\lceil}{\rceil}
\DeclarePairedDelimiter\floor{\lfloor}{\rfloor}
\DeclarePairedDelimiter\parenv{\lparen}{\rparen}

\DeclarePairedDelimiter\set{\{}{\}}

\renewcommand{\leq}{\leqslant}

\renewcommand{\geq}{\geqslant}

\newcommand{\cB}{\mathcal{B}}

\newcommand{\bfa}{\mathbf{a}}
\newcommand{\bfb}{\mathbf{b}}
\newcommand{\bfc}{\mathbf{c}}

\newcommand{\bfe}{\mathbf{e}}
\newcommand{\bfh}{\mathbf{h}}
\newcommand{\bfi}{\mathbf{i}}
\newcommand{\bfj}{\mathbf{j}}
\newcommand{\bfl}{\boldsymbol{\ell}}

\newcommand{\bfs}{\mathbf{s}}

\newcommand{\bfy}{\mathbf{y}}
\newcommand{\bfz}{\mathbf{z}}
\newcommand{\bfzero}{\mathbf{0}}
\newcommand{\bfone}{\mathbf{1}}

\newcommand{\eps}{\varepsilon}
\newcommand{\bfeps}{\boldsymbol{\varepsilon}}

\newcommand{\bfDelta}{\boldsymbol{\Delta}}

\newcommand{\F}{\mathbb{F}}

\newcommand{\N}{\mathbb{N}}

\newcommand{\Z}{\mathbb{Z}}

\newcommand{\eqdef}{\triangleq}

\DeclareMathOperator{\wt}{wt}

\newcommand{\T}{\intercal}
\newcommand{\str}{\mathrm{str}}

\DeclareMathOperator{\ord}{ord}

\begin{document}

\title[]{On Multidimensional 2-Weight-Limited Burst-Correcting Codes}


\author[1]{\fnm{Hagai} \sur{Berend}}\email{hagaiber@post.bgu.ac.il}

\author[1]{\fnm{Ohad} \sur{Elishco}}\email{elishco@gmail.com}

\author[1,2]{\fnm{Moshe} \sur{Schwartz}}\email{schwartz.moshe@mcmaster.ca}

\affil[1]{\orgdiv{School of Electrical and Computer Engineering}, \orgname{Ben-Gurion University of the Negev}, \orgaddress{\city{Beer Sheva}, \postcode{8410501}, \country{Israel}}}

\affil[2]{\orgdiv{Department of Electrical and Computer Engineering}, \orgname{McMaster University}, \orgaddress{\city{Hamilton}, \postcode{L8S 4K1}, \state{ON}, \country{Canada}}}


\abstract{
We consider multidimensional codes capable of correcting a burst error of weight at most $2$. When two positions are in error, the burst limits their relative position. We study three such limitations: the $L_\infty$ distance between the positions is bounded, the $L_1$ distance between the positions is bounded, or the two positions are on an axis-parallel line with bounded distance between them. In all cases we provide explicit code constructions, and compare their excess redundancy to a lower bound we prove.
}




\maketitle

%%%%%%%%%%%%%%%%%%%%%%%%%%%%%%%%%%%%%%%%%%%%%%%%%%%%%%%%%%%%%%%%
%%%%%%%%%%%%%%%%%%%%%%%%%%%%%%%%%%%%%%%%%%%%%%%%%%%%%%%%%%%%%%%%
%%%%%%%%%%%%%%%%%%%%%%%%%%%%%%%%%%%%%%%%%%%%%%%%%%%%%%%%%%%%%%%%
\section{Introduction}
%%%%%%%%%%%%%%%%%%%%%%%%%%%%%%%%%%%%%%%%%%%%%%%%%%%%%%%%%%%%%%%%
%%%%%%%%%%%%%%%%%%%%%%%%%%%%%%%%%%%%%%%%%%%%%%%%%%%%%%%%%%%%%%%%
%%%%%%%%%%%%%%%%%%%%%%%%%%%%%%%%%%%%%%%%%%%%%%%%%%%%%%%%%%%%%%%%

Burst-error-correcting codes have a long history. Such codes are capable of correcting localized errors. In the classical one-dimensional setting, codewords are vectors, and a burst of length $b$ is a set of errors that occur within $b$ consecutive positions.

It is, however, the multidmensional case which brings to light the richness of burst errors. In this setting, a code $C\subseteq \F_2^{n_1\times n_2 \times \dots \times n_D}$ is a set of $D$-dimensional, $n_1\times n_2\times \dots \times n_D$ binary arrays, that can correct any number of errors (bit flips) occurring in a localized area. The definition of the shape of this local area has many variants in the literature, of which we mention a short selection. Perhaps the most commonly studied model defines a burst as a pattern of errors confined to a $D$-dimensional box of fixed size $b_1\times \dots \times b_D$, e.g.,~\cite{AbdMcETil88,BlaFar94,BreBosZyaSid98,Boy13}. Other definitions include criss-cross errors~\cite{Rot91}, axis-parallel errors~\cite{SchEtz05}, $L_1$-metric balls~\cite{EtzYaa09}, and more generally, any connected area of a prescribed volume~\cite{BlaBruVar98,EtzSchVar05,JiaCooBru06,SliBru09}.

Etzion and Yaakobi~\cite{EtzYaa09} raised the interesting problem of designing codes capable of correcting a burst of limited weight. In this setting, not only are the erroneous positions confined to a certain multidimensional burst shape, but they are also limited to a prescribed Hamming weight. They studied this in two-dimensions for $b_1\times b_2$ box-shaped bursts, and later in~\cite{YaaEtz10}, for multidimensional codes in the extreme case of errors up to weight $2$, confined to either crosses or semi-crosses.

In this work we continue studying multidimensional error-correcting codes, capable of correcting a burst of error with a weight limit of $2$. Our main contributions are as follows: We consider the various burst shapes that have been studied in the unrestricted-burst-weight case. In the $2$-weight-limited case, when two errors actually occur, they become either two positions with limited $L_\infty$ distance, two positions with limited $L_1$-metric distance, or two positions along an axis-parallel line with limited distance (the straight model). In all of these cases we provide explicit linear code constructions. We also prove lower bounds on the excess redundancy of the codes. The results are summarized in Table~\ref{tab:summary} in Section~\ref{sec:conc}. We note that when the dimension $D$ and the burst size $b$ are constant, while the codeword size is arbitrarily large, the excess redundancy of our constructions is at most a constant above the lower bound.

The construction techniques we employ are varied. In some cases, even though we construct binary codes, we use codes over $\F_q$ that can correct arbitrary errors (not bursts) in the Lee metric. In another case we use packing designs to carefully devise a parity-check matrix for our code.

The remainder of this paper is organized as follows. 
Section~\ref{sec:prelim} provides the necessary preliminaries, definitions, and notation used throughout the paper. 
Sections~\ref{sec:linf},~\ref{sec:l1}, and~\ref{sec:straight}, are dedicated to the $L_{\infty}$, $L_{1}$, and straight models, respectively. In each of these sections, we present our code constructions, prove their error-correction capabilities, analyze their excess redundancy, and derive a lower bound on the excess redundancy for any code in that model. 
Finally, Section~\ref{sec:conc} concludes the paper with a summary of our results and a discussion of their implications.

%%%%%%%%%%%%%%%%%%%%%%%%%%%%%%%%%%%%%%%%%%%%%%%%%%%%%%%%%%%%%%%%
%%%%%%%%%%%%%%%%%%%%%%%%%%%%%%%%%%%%%%%%%%%%%%%%%%%%%%%%%%%%%%%%
%%%%%%%%%%%%%%%%%%%%%%%%%%%%%%%%%%%%%%%%%%%%%%%%%%%%%%%%%%%%%%%%
\section{Preliminaries}
\label{sec:prelim}

For a positive integer $n\in\N$, we define $[n]\eqdef\set{0,1,\dots,n-1}$. Let $\F_q$ denote the finite field of size $q$. Since we shall be interested in binary codes, we shall only use $\F_2$ and its extensions, $\F_{2^m}$. We use bold lower-case letters to denote vectors and upper-case letters for matrices, e.g., $\bfa\in\F_2^n$ denotes a binary vector of length $n$, and $G\in\F_2^{k\times n}$ denotes a $k\times n$ binary matrix. We assume all vectors are row vectors, and transpose them if we need a column vector. We also use $\bfzero$ to denote the all-zero vector, and $\bfone$ to denote the all-one vector.

A binary $[N,k]$ linear code, $C$, is simply a $k$-dimensional vector space, $C\subseteq\F_2^N$, $\dim(C)=k$. We call the vectors in $C$ \emph{codewords}. Since we are interested in multi-dimensional codes, the components of each codeword are mapped to a $D$-dimensional array of size $n^{\times D}=n\times n\times \dots \times n$, with $N=n^D$. We therefore index the components of each codeword by $\bfi\in[n]^D$. Thus, a codeword $\bfc\in C$ may be written as $\bfc=(c_{\bfi})_{\bfi\in [n]^D}$. We also say that the code is $D$-dimensional with length $n^{\times D}$, and we may say it has parameters $[n^{\times D},k]$. We mention briefly that the results in this paper may be easily generalized to an asymmetric case, where the $D$-dimensional medium is indexed by $[n_0]\times [n_1]\times \dots \times [n_{D-1}]$, but we only describe the symmetric case for ease of presentation.

We shall find it useful to specify $C$ by means of a parity-check matrix $H\in\F_2^{(N-k)\times N}$ such that $C=\ker(H)$. We denote the column of $H$ at coordinate $\bfi\in [n]^D$, as $\bfh^\T_{\bfi}$. By definition, $\bfc\in C$ if and only if $H\bfc^\T=\bfzero^\T$, namely, $\sum_{\bfi\in [n]^D} c_{\bfi}\bfh^\T_{\bfi}=\bfzero^\T$. The redundancy of the code $C$ is defined as 
\[r(C)\eqdef N-k.\]
Following~\cite{AbdMcETil88,EtzYaa09}, we shall find it useful to discuss the \emph{excess redundancy} of the code $C$, which is defined as
\[\xi(C)\eqdef r(C)-\ceil{\log_2 N}.\]
As is common practice, we shall sometimes conveniently group together $m$ rows of $H$ and write them as a single row with entries from $\F_{2^m}$ (replacing each $m$ vertical binary entries with a single element from $\F_{2^m}$ using the well-known isomorphism between $\F_{2}^m$ and $\F_{2^m}$).

Assume $\bfc\in C$ is transmitted, and $\bfc+\bfe$ is received, where $\bfe\in\F_2^N$ is called the error vector. The number of errors is the (Hamming) weight of $\bfe$, defined by
\[
\wt(\bfe) \eqdef \abs*{\set*{\bfi\in [n]^D ~:~ e_{\bfi}\neq 0}}.
\]
In this work we are interested in low-weight burst errors. We recall that a burst error vector has all of its non-zero entries in some local area (which will soon be made precise). Thus, the lowest weight restriction which is non-trivial, is $2$.

We study three models of burst shapes defining the local area it occupies. We use the notation of $b$-closeness, $b\in\Z$, $b\geq 0$, defined as follows:
\begin{itemize}
\item The $L_\infty$ model: $\bfi,\bfj\in [n]^D$ are $b$-close iff $\max_{t\in[D]}\set{\abs{i_t-j_t}} < b$.
\item The $L_1$ model: $\bfi,\bfj\in [n]^D$ are $b$-close iff $\sum_{t\in[D]}\abs{i_t-j_t} < b$.
\item The straight model: $\bfi,\bfj\in [n]^D$ are $b$-close iff $\sum_{t\in[D]}\abs{i_t-j_t} < b$ and $\wt(\bfi-\bfj)=1$.
\end{itemize}

The three models of $b$-closeness are depicted in Fig.~\ref{fig:models}.

\begin{figure}
\begin{center}
\begin{overpic}[scale=0.5]
{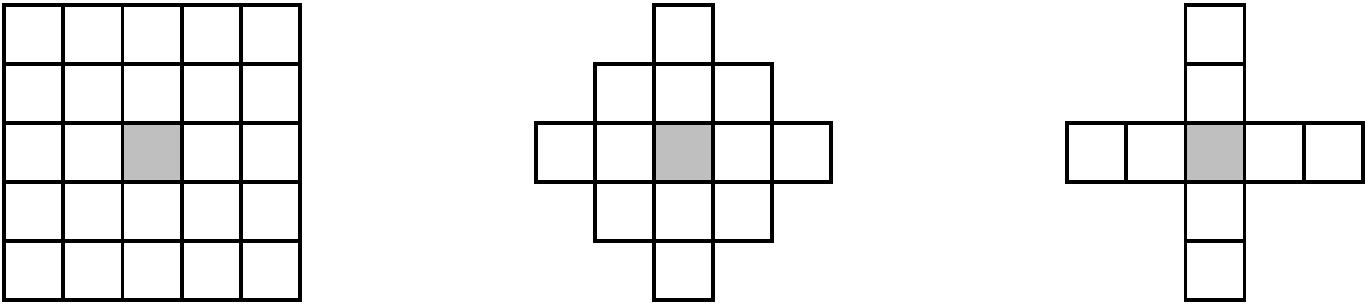}
\put(-5,10){(a)}
\put(34,10){(b)}
\put(73,10){(c)}
\end{overpic}
\end{center}
\caption{A depiction of the three $b$-closeness models. Identify each square with a point of $\Z^2$, and let $\bfi$ be the shaded square. The $2$-dimensional $3$-close positions to $\bfi$ are then the squares in (a) the $L_\infty$ model, (b) the $L_1$ model, and (c) the straight model.}
\label{fig:models}
\end{figure}

\begin{definition}
An error vector $\bfe\in\F_2^N$ is a $2$-weight-limited $b$-burst in the $L_\infty$/$L_1$/straight model, if one of the following holds:
\begin{enumerate}
\item $\wt(\bfe)\leq 1$.
\item $\wt(\bfe)=2$, $e_{\bfi}=e_{\bfj}=1$, $\bfi\neq \bfj$, and $\bfi$ and $\bfj$ are $b$-close in the $L_\infty$/$L_1$/straight model.
\end{enumerate}
\end{definition}

As a side note, we observe that in the one-dimensional case, $D=1$, the three models described above coincide, and their set of error vectors are the same.

Denote by $E_\infty(b,N)$ the set of all $2$-weight-limited $b$-bursts of length $N$ in the $L_\infty$ model. Assume $C_\infty$ is an error-correcting code for the error patterns in $E_\infty(b,N)$, defined by a parity-check matrix $H$. This is equivalent to the fact that all error patterns in $E_\infty(b,N)$ have distinct syndromes, namely, if $\bfe_1,\bfe_2\in E_{\infty}(b,N)$, $\bfe_1\neq\bfe_2$, then $H\bfe_1^\T\neq H\bfe_2^\T$. The resulting ball-packing-like argument (e.g., see~\cite{MacSlo78}) implies that
\begin{equation}
\label{eq:rlowbound}
r(C_\infty) \geq \log_2 \abs*{E_\infty(b,N)}.
\end{equation}
Analogous statements may be made for the $L_1$ model with $E_1(b,N)$, and the straight model with $E_\str(b,N)$.

In the constructions of the following sections, we make use of the following notation. Let $\bfi\in \Z^D$ be a vector of integers, $\bfi=(i_0,i_1,\dots,i_{D-1})$. We then define the notation
\[
[\bfi]_q \eqdef \sum_{t\in[D]}i_t q^t.
\]

The following proposition is folklore:

\begin{proposition}
Let $q\geq 2$ be an integer.
\label{prop:baseq}
\begin{enumerate}
\item
If $\bfi\in[q]^D$ then
\[
0 \leq [\bfi]_q \leq q^D-1,
\]
and $\bfi\mapsto [\bfi]_q$ is a bijection between $[q]^D$ and $[q^D]$.
\item
If $\bfi\in\set{-(q-1),\dots,q-1}^D$ then
\[
\abs*{[\bfi]_q}\leq q^D-1.
\]
Additionally, $[\bfi]_q=0$ if and only if $\bfi=\bfzero$.
\end{enumerate}
\end{proposition}
\begin{proof}
For the first claim, if $i_t\in[q]$ for all $t\in[D]$, then $[\bfi]_q$ is simply the value associated with the base-$q$ representation $\bfi$. In that case, it is well known that distinct $\bfi$ result in distinct $[\bfi]_q$, i.e., each value has a unique base-$q$ representation. Also, in this case,
\[
0 \leq [\bfi]_q \leq (q-1)\sum_{t=0}^{D-1} q^t = q^D-1.
\]
For the second claim, if we now allow $i_t\in\set{-(q-1),\dots,q-1}$ for all $t\in[D]$, then we can say
\[
\abs*{[\bfi]_q}\leq (q-1)\sum_{t=0}^{D-1} q^t = q^D-1.
\]
We obviously have $[\bfzero]_q=0$. We now show $\bfzero$ is the only vector with this property. Assume $\bfi\neq\bfzero$, and let $t'$ be the largest index for which $i_{t'}\neq 0$. Then
\[
\abs*{i_{t'}q^{t'}} \geq q^{t'} > (q-1)\frac{q^{t'}-1}{q-1} \geq \abs*{\sum_{t=0}^{t'-1} i_t q^t},
\]
namely, the value contributed by the largest index with non-zero entry in $\bfi$ dominates all the rest. Thus, $\bfi\neq\bfzero$ implies $[\bfi]_q\neq 0$.
\end{proof}

Another useful notation is applying the same operation to all the vector elements. For example
\[
\bfi \bmod m \eqdef (i_0 \bmod m, i_1 \bmod m, \dots , i_{D-1}\bmod m).
\]
We extend this notation convention to other operations, e.g., $\floor{\bfi/m}$ denotes the vector whose entries are the entries of $\bfi$ divided by $m$ and rounded down.

%%%%%%%%%%%%%%%%%%%%%%%%%%%%%%%%%%%%%%%%%%%%%%%%%%%%%%%%%%%%%%%%
%%%%%%%%%%%%%%%%%%%%%%%%%%%%%%%%%%%%%%%%%%%%%%%%%%%%%%%%%%%%%%%%
%%%%%%%%%%%%%%%%%%%%%%%%%%%%%%%%%%%%%%%%%%%%%%%%%%%%%%%%%%%%%%%%

\section{The $L_\infty$ Model}
\label{sec:linf}

In this section we provide a construction for codes that correct a $2$-weight-limited $b$-burst in the $L_\infty$ model. We first give a general construction, and then a slightly more restricted one with a lower excess redundancy.

\begin{construction}
\label{con:linf}
Let $n\geq b\geq 2$, and $D\geq 1$ be integers. Define
\begin{align*}
m & \eqdef \ceil*{\log_2 (n^D+1)} & a &\eqdef \ceil*{\log_2 (b^D+1)}.
\end{align*}
Let $\alpha\in \F_{2^m}$ and $\beta\in \F_{2^a}$ be primitive elements in their respective fields. Let $H$ be a matrix, whose columns are indexed by $[n]^D$, and whose $\bfi$-th column is
\begin{equation}
\label{eq:Hlinf}
\bfh^\T_\bfi = \begin{pmatrix}
\beta^{[\bfi \bmod b]_b} \\
\beta^{3[\bfi \bmod b]_b} \\
\floor{\bfi^\T / b} \bmod 2 \\
\alpha^{[\bfi]_n}
\end{pmatrix}.
\end{equation}
We construct the $D$-dimensional binary code $C$, whose parity-check matrix is $H$.
\end{construction}

\begin{theorem}
\label{th:linf}
Let $C$ be the code with parity-check matrix $H$ from Construction~\ref{con:linf}. Then $C$ is an $[N=n^D,k]$ $D$-dimensional code of length $n^{\times D}$, capable of correcting a single $2$-weight-limited $b$-burst in the $L_\infty$ model.
\end{theorem}

\begin{proof}
We prove the claim by providing a decoding algorithm. Assume $\bfc\in C$ was transmitted, and $\bfy=\bfc+\bfe$ was received, where $\bfe\in E_\infty(b,N)$. The receiver computes the syndrome:
\[
\bfs^\T=H\bfy^T=H\bfe^T=\begin{pmatrix}
s_0 \\ s_1 \\ \bfs_2^\T \\ s_3
\end{pmatrix},
\]
where the components correspond to those of $H$ from~\eqref{eq:Hlinf}, i.e., $s_0,s_1\in\F_{2^a}$, $\bfs_2\in \F_2^D$, and $s_3\in\F_{2^m}$. We distinguish between three cases, depending on the weight of $\bfe$. We show the receiver can identify the case (by computing the syndrome), and correct the errors.

\textbf{Case 1:} $\wt(\bfe)=0$. In that case, no error occurred, and we have $\bfe=\bfzero$ as well as $\bfs=\bfzero$. The decoding procedure simply returns $\bfy$ as the decoded codeword.

\textbf{Case 2:} $\wt(\bfe)=1$. In that case a single error occurred, and $\bfe$ is the all-zero vector except for the error position, $\bfi\in[n]^D$, for which $e_{\bfi}=1$. In this case, $\bfs=\bfh_{\bfi}$, the column of $H$ in position $\bfi$. This is distinguishable from Case 1 since $\bfs\neq \bfzero$. Additionally, we note that $s_1=s_0^3$, which can never happen in Case 3, hence making the cases distinguishable.

We contend that all the columns of $H$ are distinct, by virtue of the bottom element, $\alpha^{[\bfi]_n}$. To see this, by Proposition~\ref{prop:baseq}, distinct $\bfi\in[n]^D$ result in distinct $[\bfi]_n$, and
\[
0 \leq [\bfi]_n \leq n^D-1 < 2^m-1= \ord(\alpha),
\]
where $\ord(\alpha)$ is the multiplicative order of $\alpha$. Hence, distinct $\bfi$ result in distinct $\alpha^{[\bfi]_n}$. Thus, by comparing $\bfs$ against the columns of $H$, the error position $\bfi$ is obtained by the receiver which proceeds to flip the bit in position $\bfi$ of $\bfy$ to decode correctly.

\textbf{Case 3:} $\wt(\bfe)=2$. In this case two errors occurred. We denote their positions as $\bfi,\bfj\in[n]^D$, $\bfi\neq\bfj$, and we have
\[
\bfs^\T = \begin{pmatrix}
s_0 \\ s_1 \\ \bfs_2^\T \\ s_3
\end{pmatrix}
=\begin{pmatrix}
\beta^{[\bfi \bmod b]_b}+\beta^{[\bfj \bmod b]_b} \\
\beta^{3[\bfi \bmod b]_b}+\beta^{3[\bfj \bmod b]_b} \\
(\floor{\bfi^\T / b}+\floor{\bfj^\T / b}) \bmod 2 \\
\alpha^{[\bfi]_n}+\alpha^{[\bfj]_n}
\end{pmatrix}.
\]

By the definition of the $L_\infty$ model, $\abs{i_t-j_t}<b$ for all $t\in[D]$. It follows that $\bfi$ and $\bfj$ are contained in a translated $b^{\times D}$ cube, i.e., $\bfi,\bfj\in [b]^D+\bfz$, for some $\bfz\in[n]^D$. We now observe that, by construction, the top two components of any column of $H$, in positions forming a $b^{\times D}$ cube, are
\[
\set*{\begin{pmatrix}
\beta^{[\bfl \bmod b]_b} \\
\beta^{3[\bfl \bmod b]_b}
\end{pmatrix}}_{\bfl\in [b]^D+\bfz}
=
\set*{
\begin{pmatrix}
\beta^t\\
\beta^{3t}
\end{pmatrix}
}_{t\in [b^D]}.
\]
These, however, are just $b^D$ distinct columns of a $2$-error-correcting binary primitive BCH code~\cite[Chapter 9]{MacSlo78}. Since this BCH code is capable of correcting two errors, as well as a single error, we have
\begin{align*}
\begin{pmatrix}
\beta^{[\bfi \bmod b]_b}+\beta^{[\bfj \bmod b]_b} \\
\beta^{3[\bfi \bmod b]_b}+\beta^{3[\bfj \bmod b]_b}
\end{pmatrix}
&\neq
\begin{pmatrix}
\beta^{[\bfl \bmod b]_b} \\
\beta^{3[\bfl \bmod b]_b}
\end{pmatrix},
&
\begin{pmatrix}
\beta^{[\bfi \bmod b]_b}+\beta^{[\bfj \bmod b]_b} \\
\beta^{3[\bfi \bmod b]_b}+\beta^{3[\bfj \bmod b]_b}
\end{pmatrix}
& \neq
\bfzero^\T,
\end{align*}
for all $\bfl$. Thus, the receiver can distinguish between Cases 1, 2, and 3. Additionally, the receiver knows $(\bfi\bmod b)$ and $(\bfj \bmod b)$ by applying the decoding procedure for the BCH code and reversing the $[\cdot]_b$ mapping.

The next step is determining $\bfDelta\eqdef \bfj-\bfi$. Consider a partition of $\Z$ into translates of $[b]$, i.e., sets of the form $tb+[b]$, $t\in\Z$. We call each such set a block. Now, for each $t\in[D]$, $i_t$ and $j_t$ might be in the same block, in which case,
\[\Delta_t=j_t-i_t = (j_t \bmod b) - (i_t \bmod b).\]
If $i_t$ and $j_t$ are not in the same block, they must be in adjacent blocks in order for $\abs{i_t-j_t}<b$ to hold. In that case, we must have $i_t\not\equiv j_t \pmod{b}$, and the only way possible is to have
\[
\Delta_t=j_t-i_t=\begin{cases}
(j_t\bmod b)-(i_t\bmod b)+b & \text{if $(j_t\bmod b)-(i_t\bmod b)<0$,}\\
(j_t\bmod b)-(i_t\bmod b)-b & \text{if $(j_t\bmod b)-(i_t\bmod b)>0$.}\\
\end{cases}
\]
Whether $i_t$ and $j_t$ are in the same block or in adjacent blocks, may be obtained by examining the third part of the syndrome,
\[
\bfs_2=(s_{2,0},s_{2,1},\dots,s_{2,D-1})=(\floor{\bfi^\T / b}+\floor{\bfj^\T / b}) \bmod 2.
\]
By construction, $i_t$ and $j_t$ are in the same block if and only if $s_{2,t}=0$, for all $t\in[D]$. We therefore summarize
\[
\Delta_t=\begin{cases}
(j_t\bmod b)-(i_t\bmod b) & \text{if $s_{2,t}=0$,}\\
(j_t\bmod b)-(i_t\bmod b)+b & \text{if $s_{2,t}=1$ and $(j_t\bmod b)-(i_t\bmod b)<0$,}\\
(j_t\bmod b)-(i_t\bmod b)-b & \text{if $s_{2,t}=1$ and $(j_t\bmod b)-(i_t\bmod b)>0$,}\\
\end{cases}
\]
for all $t\in[D]$.

At this point, the receiver knows that
\begin{align}
\label{eq:bfidelta}
\bfi &= b\bfz+(\bfi \bmod b) &
\bfj &= b\bfz+(\bfi \bmod b) + \bfDelta,
\end{align}
with the only unknown being $\bfz\in [\ceil{\frac{n}{b}}]^D$. Using the last part of the syndrome, we can now write
\begin{align*}
s_3 &= \alpha^{[\bfi]_n}+\alpha^{[\bfj]_n} = 
\alpha^{[b\bfz + (\bfi \bmod b)]_n} + \alpha^{[b\bfz + (\bfi \bmod b)+\bfDelta]_n} \\
&= \alpha^{[b\bfz]_n} \cdot \alpha^{[\bfi \bmod b]_n}\cdot\parenv*{1+\alpha^{[\bfDelta]_n}}
\end{align*}
We now observe that $\bfDelta=\bfi-\bfj\neq\bfzero$, since $\bfi\neq\bfj$. By definition, $\bfDelta\in\set{-(b-1),\dots,b-1}^D$. Thus, by Proposition~\ref{prop:baseq},
\[
\abs*{[\bfDelta]_n} \leq n^D-1 < 2^m-1= \ord(\alpha).
\]
Additionally, since $\bfDelta\neq\bfzero$, by Proposition~\ref{prop:baseq}, $[\bfDelta]_n\neq 0$. It now follows that
\[
1+\alpha^{[\bfDelta]_n}\neq 0.
\]
The receiver now needs to solve
\begin{equation}
\label{eq:bfz}
\alpha^{[b\bfz]_n}=\frac{s_3}{\alpha^{[\bfi\bmod b]_n}\parenv*{1+\alpha^{[\bfDelta]_n}}}.
\end{equation}
Since $\bfz\in[\ceil{\frac{n}{b}}]^D$, by Proposition~\ref{prop:baseq},
\[ 0\leq [b\bfz]_n \leq n^D-1<2^m-1=\ord(\alpha),\]
and so a unique $\bfz\in[\ceil{\frac{n}{b}}]^D$ exists to solve~\eqref{eq:bfz}. The receiver now knows $\bfz$, $(\bfi \bmod b)$, and $\bfDelta$, and so by~\eqref{eq:bfidelta}, the receiver knows the two positions in error, $\bfi$ and $\bfj$. It can then flip the bits in these two positions in $\bfy$ and correctly recover the transmitted codeword $\bfc$.
\end{proof}

\begin{corollary}
\label{cor:linfxi}
Let $C$ be the code from Construction~\ref{con:linf}. Then its excess redundancy is
\[
\xi(C)=\begin{cases}
2\ceil{\log_2(b^D+1)}+D+1 & \text{$n$ is a power of $2$,}\\
2\ceil{\log_2(b^D+1)}+D & \text{otherwise.}
\end{cases}
\]
\end{corollary}
\begin{proof}
In the notation of Construction~\ref{con:linf}, by definition,
\begin{align*}
\xi(C) &= 2a+D+m-\ceil*{\log_2(n^D)} \\
&= 2\ceil*{\log_2(b^D+1)}+D+\ceil*{\log_2(n^D+1)}-\ceil*{\log_2(n^D)},
\end{align*}
and the difference between the last two terms is $0$ or $1$ depending on whether $n$ is a power of $2$, thus proving the claim.
\end{proof}

In certain cases, we can improve upon Construction~\ref{con:linf} by extending the length of the code from $n^{\times D}$ arrays to $(bn)^{\times D}$ arrays. By increasing the length, we reduce the excess redundancy of the code. The construction remains essentially the same, as does its proof of correctness.

\begin{construction}
\label{con:linfext}
Let $n\geq b\geq 2$, and $D\geq 1$ be integers. Define
\begin{align}
\label{eq:linfm}
m & \eqdef \ceil*{\log_2 (n^D+1)} & a &\eqdef \ceil*{\log_2 (b^D+1)}.
\end{align}
Assume $\gcd(b,2^m-1)=1$. Let $\alpha\in \F_{2^m}$ and $\beta\in \F_{2^a}$ be primitive elements in their respective fields. Let $H$ be a matrix, whose columns are indexed by $[bn]^D$, and whose $\bfi$-th column is
\[
\bfh^\T_\bfi = \begin{pmatrix}
\beta^{[\bfi \bmod b]_b} \\
\beta^{3[\bfi \bmod b]_b} \\
\floor{\bfi^\T / b} \bmod 2 \\
\alpha^{[\bfi]_n}
\end{pmatrix}.
\]
We construct the $D$-dimensional binary code $C$, whose parity-check matrix is $H$.
\end{construction}

\begin{theorem}
\label{th:linfext}
Let $C$ be the code with parity-check matrix $H$ from Construction~\ref{con:linfext}. Then $C$ is an $[N=(bn)^D,k]$ $D$-dimensional code of length $(bn)^{\times D}$, capable of correcting a single $2$-weight-limited $b$-burst in the $L_\infty$ model.
\end{theorem}

\begin{proof}
The proof proceeds almost exactly as that of Theorem~\ref{th:linf}, and so we highlight only the subtle differences.

In Case 2, we need to show that the columns of $H$ are all distinct. Assume $\bfh^\T_{\bfi}=\bfh^\T_{\bfj}$, for some $\bfi,\bfj\in[bn]^D$. By the top element of each column, we have
\[
\beta^{[\bfi \bmod b]_b}=\beta^{[\bfj \bmod b]_b}.
\]
Since, by Proposition~\ref{prop:baseq},
\[ 0\leq [\bfi \bmod b]_b, [\bfj\bmod b]_b \leq b^D-1 < 2^a-1=\ord(\beta),\]
we must therefore have
\[ \bfi \bmod b = \bfj \bmod b,\]
implying that $\bfj = \bfi + b\bfz$ for some $\bfz\in\set{-(n-1),\dots,n-1}^D$. Moving on to check the bottom element of each column, we now have
\[
\alpha^{[\bfi]_n}=\alpha^{[\bfj]_n}=\alpha^{[\bfi]_n+b[\bfz]_n},
\]
and so
\begin{equation}
\label{eq:alphabz}
(\alpha^b)^{[\bfz]_n}=1.
\end{equation}
Since $\gcd(b,2^m-1)=1$, $\alpha^b$ is also a primitive element in $\F_{2^m}$. By Proposition~\ref{prop:baseq},
\[
\abs*{[\bfz]_n}\leq n^D-1 < 2^m-1= \ord(\alpha)=\ord(\alpha^b).
\]
Thus, the only way for~\eqref{eq:alphabz} to hold is to have $\bfz=\bfzero$, and then $\bfi=\bfj$, proving our claim that the columns of $H$ are distinct.

In Case 3, when reaching~\eqref{eq:bfz}, we get
\[
\alpha^{[b\bfz]_n}=(\alpha^b)^{[\bfz]_n}=\frac{s_3}{\alpha^{[\bfi\bmod b]_n}\parenv*{1+\alpha^{[\bfDelta]_n}}},
\]
with $\bfz\in[n]^D$, and we need to show there is a unique $\bfz$ solving it. Recall that $\alpha^b$ is also primitive. By Proposition~\ref{prop:baseq},
\[
0\leq [\bfz]_n\leq n^D-1 < 2^m-1= \ord(\alpha^b),
\]
and so $\bfz$ is unique.
\end{proof}

\begin{corollary}
\label{cor:linfxiext}
Let $C$ by the code from Construction~\ref{con:linfext}. Then its excess redundancy satisfies
\[
\xi(C) \leq D\log_2 (2b) + 3.
\]
\end{corollary}
\begin{proof}
By Construction,
\[
\xi(C) = 2\ceil*{\log_2(b^D+1)}+D+\ceil*{\log_2(n^D+1)}-\ceil*{\log_2\parenv*{(bn)^D}}.
\]
The claim follows by using the simple $\ceil{\log_2(x+1)}\leq \log_2 x + 1$ for integer $x$, as well as $x \leq \ceil{x}\leq x+1$ for all real $x$.
\end{proof}

In order for us to assess the efficiency of our constructions, we provide a lower bound on the excess redundancy for all codes in this model.

\begin{theorem}
\label{th:linfxilower}
Fix integers $D\geq 1$, and $b\geq 2$. Let $C$ be an $[N=n^D,k]$ $D$-dimensional code of length $n^{\times D}$, $n\geq (2b-1)^{D-1}D(b^2-b)$, that is capable of correcting a single $2$-weight-limited $b$-burst in the $L_\infty$ model. Then its excess redundancy is lower bounded by
\[
\xi(C) \geq D\log_2(2b-1)-2.
\]
\end{theorem}

\begin{proof}
By the ball-packing argument of~\eqref{eq:rlowbound}
\[
\xi(C) \geq \log_2 \abs*{E_\infty(b,N)}-\ceil*{\log_2 N}.
\]
To lower bound $\abs{E_\infty(b,N)}$, we observe that this set contains the unique all-zero array, $N=n^D$ errors of weight $1$, as well error patterns of weight $2$.

To count the latter, we first choose position $\bfi$ for the first error, and then position $\bfj$ for the second error. For each $\ell\in D$ we choose $i_\ell$ and then $j_\ell$. If we choose $b-1\leq i_\ell \leq n-b$, there are $2b-1$ ways of choosing $j_\ell$ such that $\abs{i_\ell-j_\ell}\leq b-1$. If $0\leq i_\ell\leq b-2$, there are $b+i_\ell$ for choosing $j_\ell$, and a similar counting holds for $n-b+1\leq i_\ell\leq n-1$. We repeat this for each $\ell\in[D]$. From this counting we need to subtract $n^D$ cases in which we end up choosing $\bfi=\bfj$, and then divide by $2$ since we can switch the order of $\bfi$ and $\bfj$. Thus, in total
\begin{align*}
\abs*{E_\infty(b,N)} &= 1+n^D+\frac{1}{2}\parenv*{\parenv*{(n-2b+2)(2b-1)+2\sum_{s=0}^{b-2}(b+s)}^D-n^D} \\
&= 1 + n^D+\frac{1}{2}\parenv*{(2nb-n-b^2+b)^D-n^D}\\
&\geq \frac{1}{2}\parenv*{(2nb-n-b^2+b)^D+n^D} \\
&= \frac{1}{2}n^D\parenv*{(2b-1)^D\parenv*{1-\frac{b^2-b}{n(2b-1)}}^D+1} \\
&\overset{(a)}{\geq} \frac{1}{2}n^D\parenv*{(2b-1)^D\parenv*{1-D\frac{b^2-b}{n(2b-1)}}+1}\\
&= \frac{1}{2}n^D\parenv*{(2b-1)^D-(2b-1)^{D-1}D\frac{b^2-b}{n}+1} \\
&\overset{(b)}{\geq} \frac{1}{2}n^D(2b-1)^D,
\end{align*}
where $(a)$ follows from Bernoulli's Inequality,
\[
(1+x)^r \geq 1+rx, \qquad \text{for all real $x\geq -1$ and $r\in\N$,}
\]
and $(b)$ follows for all $n\geq (2b-1)^{D-1}D(b^2-b)$.

Plugging this into the bound on $\xi(C)$ we get
\[
\xi(C) \geq \log_2\parenv*{\frac{1}{2}n^D(2b-1)^D} - \ceil*{\log_2 n^D} \geq D\log_2(2b-1)-2,
\]
by using $\ceil{x}\leq x+1$.
\end{proof}

By comparing the lower bound on the excess redundancy provided by Theorem~\ref{th:linfxilower} and the excess redundancy of the constructed codes in Corollary~\ref{cor:linfxi} and Corollary~\ref{cor:linfxiext}, we see that the codes have redundancy that very closely matches the lower bound. As a final comment for this section, we mention that when $b$ is a power of $2$, we can further reduce the code redundancy by one bit. This is accomplished by slightly redesigning the parity-check matrix, replacing the two-error-correcting BCH code, with an extended two-error-correcting BCH code. In that case, we can define the columns of the parity-check matrix to be
\[
\bfh^\T_\bfi = \begin{pmatrix}
1 \\
\beta_{[\bfi \bmod b]_b} \\
\beta_{[\bfi \bmod b]_b}^3 \\
\floor{\bfi^\T / b} \bmod 2 \\
\alpha^{[\bfi]_n}
\end{pmatrix}.
\]
where $\beta_0,\dots,\beta_{2^a-1}$ are the $2^a$ elements of $\F_{2^a}$, $a=\log_2 b^D=D\log_2 b$ (compare this to the original constructions' requirement of $a=\ceil{\log_2 (b^D+1)}$). The proofs for correctness are the same, and we omit them.

%%%%%%%%%%%%%%%%%%%%%%%%%%%%%%%%%%%%%%%%%%%%%%%%%%%%%%%%%%%%%%%%
%%%%%%%%%%%%%%%%%%%%%%%%%%%%%%%%%%%%%%%%%%%%%%%%%%%%%%%%%%%%%%%%
%%%%%%%%%%%%%%%%%%%%%%%%%%%%%%%%%%%%%%%%%%%%%%%%%%%%%%%%%%%%%%%%
\section{The $L_1$ Model}
\label{sec:l1}
%%%%%%%%%%%%%%%%%%%%%%%%%%%%%%%%%%%%%%%%%%%%%%%%%%%%%%%%%%%%%%%%
%%%%%%%%%%%%%%%%%%%%%%%%%%%%%%%%%%%%%%%%%%%%%%%%%%%%%%%%%%%%%%%%
%%%%%%%%%%%%%%%%%%%%%%%%%%%%%%%%%%%%%%%%%%%%%%%%%%%%%%%%%%%%%%%%

The next model we study is the $L_1$ model. The solution we propose for this model is quite involved. In particular, the top two layers of the parity-check matrix form a Lee-metric code.

We first recall the definition of the Lee metric. Let $p$ be a prime, and consider the finite field $\F_p=\set{0,1,\dots,p-1}$. For each element in $\alpha\in\F_p$ we define the Lee value as
\[
\abs*{\alpha}_L = \min\set{\alpha,p-\alpha}\in \Z.
\]
Then, for a vector $\bfa\in\F_p^n$, we define the Lee weight as
\[
\wt_L(\bfa) = \sum_{i\in[n]} \abs{a_i}_L,
\]
where the summation is over the integers. Finally, the Lee distance between two vectors, $\bfa,\bfb\in\F_p^n$ is simply
\[
d_L(\bfa,\bfb)=\wt_L(\bfa-\bfb).
\]
For our construction, we need the following result first.

\begin{lemma}[{{\cite{RotSie94}}}]
\label{lem:bch}
Let $p,b,n,s$ be positive integers, where $p$ is a prime, $p\geq 2b+1$, and $p^s\geq n+1$. Let $C$ be a BCH code of length $n$ over $\F_p$ with parity-check matrix over $\F_{p^s}$
\[
H = \begin{pmatrix}
1 & 1 & \dots & 1 \\
\alpha_0 & \alpha_1 & \dots & \alpha_{n-1} \\
\alpha_0^2 & \alpha_1^2 & \dots & \alpha_{n-1}^2 \\
\vdots & \vdots & & \vdots \\
\alpha_0^{b-1} & \alpha_1^{b-1} & \dots & \alpha_{n-1}^{b-1}
\end{pmatrix},
\]
where $\alpha_0,\dots,\alpha_{n-1}\in\F_{p^s}$ are distinct code roots, with at least $b$ of them being consecutive roots. Then $C$ has Lee distance at least $2b$, and can therefore correct $b-1$ errors in the Lee metric. Additionally, the redundancy of $C$ satisfies
\[
r(C) \leq 1 + (b-1)s.
\]
\end{lemma}

With this lemma, we can now describe our construction.

\begin{construction}
\label{con:l1}
Let $n\geq b\geq 2$ and $D\geq 1$ be integers. Let $p$ be the smallest prime such that $p\geq 2b+1$. Define
\[
s \eqdef \ceil*{\log_p (D+1)},
\]
and $C_L$ be a BCH code, of length $D$, over $\F_p$, correcting $b-1$ errors in the Lee metric, as described in Lemma~\ref{lem:bch}. Let $A$ be an $r(C_L)\times D$ matrix over $\F_p$ that is a parity-check matrix for $C_L$. Further define
\begin{align}
\label{eq:l1m}
m & \eqdef \ceil*{\log_2 (n^D+1)}, & 
a &\eqdef \ceil*{\log_2 (p^{r(C_L)}+1)}, &
a' &\eqdef \ceil*{\log_2 (p^D+1)}.
\end{align}
Let $\alpha\in \F_{2^m}$, $\beta\in \F_{2^a}$, and $\gamma\in\F_{2^{a'}}$, be primitive elements in their respective fields. Assume $\gcd(p,2^m-1)=1$. Let $H$ be a matrix, whose columns are indexed by $[np]^D$, and whose $\bfi$-th column is
\begin{equation}
\label{eq:Hl1}
\bfh^\T_\bfi = \begin{pmatrix}
\beta^{[A\bfi^\T \bmod p]_p} \\
\beta^{3[A\bfi^\T \bmod p]_p} \\
\floor{\bfi^\T / p}  \bmod 2 \\
\gamma^{[\bfi \bmod p]_p} \\
\alpha^{[\bfi]_n}
\end{pmatrix}.
\end{equation}
We construct the $D$-dimensional binary code $C$, whose parity-check matrix is $H$.
\end{construction}

\begin{theorem}
Let $C$ be the code with parity-check matrix $H$ from Construction~\ref{con:l1}. Then $C$ is an $[N=(np)^D,k]$ $D$-dimensional code of length $(np)^{\times D}$, capable of correcting a single $2$-weight-limited $b$-burst in the $L_1$ model.
\end{theorem}

\begin{proof}
The proof is an elaboration of the process presented in the proof of Theorem~\ref{th:linf}. Assume $\bfc\in C$ was transmitted, and $\bfy=\bfc+\bfe$ was received, where $\bfe\in E_1(b,N)$. The receiver computes the syndrome:
\[
\bfs^\T=H\bfy^\T=H\bfe^\T=\begin{pmatrix}
s_0 \\
s_1 \\
\bfs^\T_2 \\
s_3 \\
s_4
\end{pmatrix}
\]
where the components correspond to those of $H$ from~\eqref{eq:Hl1}, i.e., $s_0,s_1\in\F_{2^a}$, $\bfs_2\in\F^D_2$, $s_3\in\F_{2^{a'}}$, and $s_4\in\F_{2^m}$. We distinguish between cases depending on $\wt(\bfe)$.

\textbf{Case 1:} $\wt(\bfe)=0$. In this case $\bfe=\bfzero$ and $\bfs=\bfzero$ and the decoder returns $\bfy$.

\textbf{Case 2:} $\wt(\bfe)=1$. Let $\bfi\in[np]^D$ be the location of the error, i.e., the unique index such that $e_{\bfi}=1$. Thus, $\bfs=\bfh_{\bfi}$. The decoder distinguishes this case from Case 1 by the fact that $s_0,s_1\neq 0$. Additionally, $s_0^3=s_1$ which we later show cannot happen in Case 3.

To complete this case, it remains to show that the columns of $H$ are all distinct. Assume $\bfh^\T_{\bfi}=\bfh^\T_{\bfj}$, for some $\bfi,\bfj\in[np]^D$. By considering $s_3$, we have
\[
\gamma^{[\bfi \bmod p]_p}=\gamma^{[\bfj \bmod p]_p}.
\]
Since, by Proposition~\ref{prop:baseq},
\[ 0\leq [\bfi \bmod p]_p, [\bfj \bmod p]_p \leq p^D-1 < 2^{a'}-1=\ord(\gamma),\]
we must therefore have,
\[ \bfi \equiv \bfj \pmod{p}.\]
This implies that $\bfj = \bfi + p\bfz$ for some $\bfz\in\set{-(n-1),\dots,n-1}^D$. From the last element of the syndrome we now have
\[
\alpha^{[\bfi]_n}=\alpha^{[\bfj]_n}=\alpha^{[\bfi]_n+p[\bfz]_n},
\]
and so
\begin{equation}
\label{eq:l1alphabp}
(\alpha^p)^{[\bfz]_n}=1.
\end{equation}
Since $\gcd(p,2^m-1)=1$, $\alpha^p$ is also a primitive element in $\F_{2^m}$. By Proposition~\ref{prop:baseq},
\[
\abs*{[\bfz]_n}\leq n^D-1 < 2^m-1= \ord(\alpha)=\ord(\alpha^p).
\]
Thus, the only way for~\eqref{eq:l1alphabp} to hold is to have $\bfz=\bfzero$, and then $\bfi=\bfj$, proving our claim that the columns of $H$ are distinct.

\textbf{Case 3:} $\wt(\bfe)=2$. We denote the two erroneous positions by $\bfi,\bfj\in[np]^D$, $\bfi\neq\bfj$. We have
\[
\bfj = \bfi + \bfeps, \qquad\text{ with } \sum_{\ell\in[D]} \abs*{\eps_\ell} \leq b-1.
\]
We obviously also have
\[
\bfj \equiv \bfi + \bfeps \pmod{p},
\]
but crucially,
\begin{equation}
\label{eq:leeeps}
\wt_L(\bfeps \bmod p) \leq b-1,
\end{equation}
thereby translating the $L_1$ constraint to a Lee-metric constraint.

The receiver can compute the syndrome, obtaining
\[
\bfs^\T = \begin{pmatrix}
s_0 \\ s_1 \\ \bfs^\T_2 \\ s_3 \\ s_4
\end{pmatrix}
=\begin{pmatrix}
\beta^{[A\bfi^\T \bmod p]_p}+\beta^{[A\bfj^\T \bmod p]_p} \\
\beta^{3[A\bfi^\T \bmod p]_p}+\beta^{3[A\bfj^\T \bmod p]_p} \\
(\floor{\bfi^\T / p} + \floor{\bfj^\T / p} ) \bmod 2 \\
\gamma^{[\bfi \bmod p]_p}+\gamma^{[\bfj \bmod p]_p} \\
\alpha^{[\bfi]_n}+\alpha^{[\bfj]_n}
\end{pmatrix}.
\]

By construction, $A$ is a parity-check matrix for a (linear shortened) BCH code over $\F_p$ that is capable of correcting error patterns of Lee weight at most $b-1$. Thus, by~\eqref{eq:leeeps},
\[
A\bfj^\T \equiv A(\bfi + \bfeps)^\T \not\equiv A\bfi^\T \pmod{p}.
\]
The cosets of $C_L$ are in one-to-one relation with the syndromes defined by $A$. Thus, each syndrome $\bfz\in\F_p^{r(C_L)}$ (eq., coset) is assigned a unique value $[\bfz]_p$ by Proposition~\ref{prop:baseq}. Since $p^{r(C_L)}-1<2^a-1\leq \ord(\beta)$, distinct syndromes are assigned distinct $\beta^{[\bfz]_p}$. It follows that
\[
\begin{pmatrix}
\beta^{[A\bfi^\T \bmod p]_p} \\
\beta^{3[A\bfi^\T \bmod p]_p} \\
\end{pmatrix}
\qquad \text{and} \qquad
\begin{pmatrix}
\beta^{[A\bfj^\T \bmod p]_p} \\
\beta^{3[A\bfj^\T \bmod p]_p} \\
\end{pmatrix}
\]
are distinct columns of a $2$-error-correcting binary primitive BCH code~\cite[Chapter 9]{MacSlo78}. Since that code is capable of correcting two errors, as well as a single error, we have
\begin{align*}
\begin{pmatrix}
\beta^{[A\bfi^\T \bmod p]_p}+\beta^{[A\bfj^\T \bmod p]_p} \\
\beta^{3[A\bfi^\T \bmod p]_p}+\beta^{3[A\bfj^\T \bmod p]_p}
\end{pmatrix}
& \neq
\bfzero^\T, \\
\begin{pmatrix}
\beta^{[A\bfi^\T \bmod p]_p}+\beta^{[A\bfj^\T \bmod p]_p} \\
\beta^{3[A\bfi^\T \bmod p]_p}+\beta^{3[A\bfj^\T \bmod p]_p}
\end{pmatrix}
&\neq
\begin{pmatrix}
\beta^{[A\bfl^\T \bmod p]_p} \\
\beta^{3[A\bfl^\T \bmod p]_p}
\end{pmatrix},
\end{align*}
for all $\bfl\in[np]^D$. Thus, the receiver can distinguish between Cases 1, 2, and 3. Additionally, the receiver knows $[A\bfi^\T \bmod p]_p$ and $[A\bfj^\T \bmod p]_p$ by applying the decoding procedure for the binary BCH code. Since $[\cdot]_p$ is one-to-one from $\F_p^{r(C_L)}$ to $[2^a]$, the receiver now also knows $A\bfi^\T \bmod p$ and $A\bfj^\T \bmod p$.

The next step is determining $\bfeps$. Over $\F_p$, the receiver can compute
\[
A\bfj^\T-A\bfi^\T \equiv A\bfeps^\T \pmod{p}.
\]
Since $\wt_L(\bfeps \bmod p)\leq b-1$, and since $C_L$ can correct $b-1$ errors in the Lee metric, the receiver can therefore find $\bfeps \bmod p$. Let us conveniently denote $\bfeps \bmod p = (\eps'_0,\eps'_1,\dots,\eps'_{D-1})$. We would like to find the exact value of $\bfeps$, bearing in mind that we know that $\sum_{\ell\in[D]}\abs{\eps_\ell} \leq b-1$. Since $p\geq 2b+1$ there is a unique way of doing so, by setting
\[
\eps_\ell = \begin{cases}
\eps'_\ell & 0\leq \eps'_\ell \leq b-1, \\
\eps'_\ell-p & p-b+1 \leq \eps'_\ell \leq p-1,
\end{cases}
\]
for each $\ell\in[D]$. Hence, the receiver may now know $\bfeps$.

Next, we undo the multiplication by the matrix $A$. To that end, we focus on $\bfs_2=(s_{2,0},s_{2,1},\dots,s_{2,D-1})$ and $s_3$. If we examine the $\ell$-th coordinate of $\bfi$ and $\bfj$, we know that $j_\ell=i_\ell+\eps_\ell$. We would like to get a similar equation involving $j_\ell \bmod p$ and $i_\ell \bmod p$. We note that the third component of every column of $H$ (which is responsible for computing $\bfs_2$), is effectively tiling $\Z$ with translates of $[p]$ in each direction. Thus, $s_{2,\ell}$ is $0$ if $i_\ell$ and $j_\ell$ are in the same translate, and $1$ otherwise. Thus, we may write
\begin{align*}
j_\ell \bmod p &= i_\ell \bmod p + \overline{\eps_\ell},\\
\overline{\eps_\ell} &= \begin{cases}
\eps_\ell & s_{2,\ell} = 0, \\
\eps_\ell-p & \text{$s_{2,\ell} = 1$ and $\eps_\ell>0$},\\
\eps_\ell+p & \text{$s_{2,\ell} = 1$ and $\eps_\ell<0$},
\end{cases}
\end{align*}
for all $\ell\in[D]$. Additionally, $-p < \overline{\eps_\ell} < p$. Collecting these together, we define $\overline{\bfeps}=(\overline{\eps_0},\dots,\overline{\eps_{D-1}})$.
Shifting to look at $s_3$, we now have
\begin{equation}
\label{eq:l1s3}
s_3 = \gamma^{[\bfi \bmod p]_p}+\gamma^{[\bfj \bmod p]_p} = \gamma^{[\bfi \bmod p]_p}+\gamma^{[\bfi \bmod p + \overline{\bfeps}]_p}
= \gamma^{[\bfi \bmod p]_p}\parenv*{1+\gamma^{[\overline{\bfeps}]_p}}.
\end{equation}
By Proposition~\ref{prop:baseq},
\[
0 < \abs*{[\overline{\bfeps}]_p} < p^D-1 < 2^{a'}-1 = \ord(\gamma),
\]
and so
\[ 1+\gamma^{[\overline{\bfeps}]_p}\neq 0.\]
Again, by Proposition~\ref{prop:baseq},
\[
0\leq [\bfi\bmod p]_p \leq p^D-1 < 2^{a'} -1 = \ord(\gamma),
\]
and thus, the receiver can solve~\eqref{eq:l1s3} and obtain $\bfi \bmod p$.

At this point, the receiver knows that
\begin{align}
\label{eq:l1bfidelta}
\bfi &= p\bfz+(\bfi \bmod p), &
\bfj &= p\bfz+(\bfi \bmod p) + \bfeps,
\end{align}
with the only unknown being $\bfz\in [n]^D$. Using the last part of the syndrome, we can now write
\begin{align*}
s_4 &= \alpha^{[\bfi]_n}+\alpha^{[\bfj]_n} = 
\alpha^{[p\bfz + (\bfi \bmod p)]_n} + \alpha^{[p\bfz + (\bfi \bmod p)+\bfeps]_n} \\
&= \alpha^{[p\bfz]_n} \cdot \alpha^{[\bfi \bmod p]_n}\cdot\parenv*{1+\alpha^{[\bfeps]_n}}
\end{align*}
Recall that $\abs{\eps_\ell}\leq b-1$ for all $\ell\in[D]$, hence
\[
0 < \abs*{[\bfeps]_n} \leq (b-1)n^{D-1} < 2^m-1 \leq \ord(\alpha),
\]
and so
\[
1+\alpha^{[\bfeps]_n}\neq 0.
\]
The receiver now needs to solve
\begin{equation}
\label{eq:l1bfz}
\alpha^{[p\bfz]_n}=\parenv*{\alpha^p}^{[\bfz]_n}=\frac{s_4}{\alpha^{[\bfi\bmod p]_n}\parenv*{1+\alpha^{[\bfeps]_n}}}.
\end{equation}
Since $\gcd(p,2^m-1)=1$, it follows that $\alpha^p$ is also primitive, and because $\bfz\in[n]^D$, Proposition~\ref{prop:baseq} ensures that
\[
0\leq [\bfz]_n \leq n^D-1 < 2^m-1 = \ord(\alpha^p),
\]
and there is a unique solution to~\eqref{eq:l1bfz}. The receiver now knows $\bfz$, $(\bfi \bmod p)$, and $\bfeps$, and so by~\eqref{eq:l1bfidelta}, the receiver knows the two positions in error, $\bfi$ and $\bfj$. It can then flip the bits in these two positions in $\bfy$ and correctly recover the transmitted codeword $\bfc$.
\end{proof}

\begin{corollary}
\label{cor:l1xi}
Let $C$ be the code from Construction~\ref{con:l1}. Then its excess redundancy satisfies
\[
\xi(C) \leq 2b \log_2 b + 2(b-1)\log_2(D+1) + 4b+D+4.
\]
\end{corollary}
\begin{proof}
We follow the definition of the excess redundancy, and write
\begin{align*}
\xi(C) & = 2\ceil*{\log_2 (p^{r(C_L)}+1)} + D + \ceil*{\log_2 (p^D+1)} + \ceil*{\log_2 (n^D+1)} - \ceil*{\log_2 (pn)^D} \\
&\overset{(a)}{\leq} 2\log_2 p + 2(b-1)\log_2 p^s + D + 4 
\leq 2\log_2 p + 2(b-1)\log_2 (p(D+1)) + D + 4 \\
&= 2b\log_2 p + 2(b-1)\log_2 (D+1) + D + 4,
\end{align*}
where $(a)$ follows by using Lemma~\ref{lem:bch} to get $r(C_L)\leq 1+(b-1)s$, and by using $\ceil{\log_u(x+1)}\leq \log_u x + 1$ for positive integers $u,x$, as well as $x\leq \ceil{x}\leq x+1$ for all real $x$.

By construction, $p$ is the smallest prime such that $p\geq 2b+1$. Thus, using Bertrand's postulate (e.g., see~\cite[Section 22.3]{HarWri79}), $p \leq 4b-3 \leq 4b$, which gives
\[
\xi(C) \leq 2b \log_2 b + 2(b-1)\log_2(D+1) + 4b+D+4.
\]
\end{proof}

We comment that the bound of Corollary~\ref{cor:l1xi} may be quite loose due to the worst-case rounding we assume. Additionally, better results may be obtained by codes tailored for a specific value of $b$. As an example, the code from~\cite[Th.~4]{SchEtz05}, designed for $b=2$ only, has excess redundancy of $\ceil{\log_2 D}+1$ for all large enough arrays. In the following remark we describe a similar construction for $b=3$.

\begin{remark}
\label{rem:l1b3}
Let $n\geq 3$ and $D\geq 1$ be integers. Define $m\eqdef \ceil{\log_2(n^D+1)}$, and let $\alpha\in\F_{2^m}$ be primitive. Let $A$ be a $2\ceil{\log_2(D+1)}\times D$ parity-check matrix of a (shortened) binary double-error-correcting BCH code of length $D$, and let $B$ be a $\ceil{\log_2 D}\times D$ binary matrix with distinct columns. Construct a parity-check matrix, $H$, whose columns are indexed by $[n]^D$, and whose $\bfi$-th column is
\[
\bfh_i = \begin{pmatrix}
1 \\
A \cdot \bfi^\T \bmod 2 \\
\floor{\bfone \cdot \bfi^\T / 2} \bmod 2 \\
\floor{B \cdot \bfi^\T/2} \bmod 2 \\
\alpha^{[\bfi]_n}
\end{pmatrix}.
\]
Using similar techniques to previous proofs, we can show that the code with parity-check matrix $H$ is capable of correcting a single $2$-weight-limited $3$-burst in the $L_1$ model. The excess redundancy of the code is upper bounded by $3\log_2 D + 6$.
\end{remark}

For a lower bound on the excess redundancy we turn to the following theorem.

\begin{theorem}
\label{th:l1xilower}
Let $D\geq 1$ and $b\geq 2$, and $n\geq 4D(b-1)$ be integers. Let $C$ be an $[N=n^D,k]$ $D$-dimensional code of length $n^{\times D}$, that is capable of correcting a single $2$-weight-limited $b$-burst in the $L_1$ model. Then its excess redundancy satisfies
\begin{align*}
\xi(C) &\geq
\begin{cases}
b-1 + \log_2 \binom{D}{b-1} - 3 & D\geq b-1, \\
D + \log_2 \binom{b-1}{D} - 3 & D<b-1,
\end{cases} \\
&\geq
\begin{cases}
(b-1)(1+\log_2(D-b+2)-\log_2(b-1))-3 & D\geq b-1, \\
D(1+\log_2(b-D)-\log_2 D)-3 & D< b-1.
\end{cases}
\end{align*}
\end{theorem}

\begin{proof}
We use~\eqref{eq:rlowbound}, and get
\[
\xi(C) \geq \log_2 \abs*{E_1(b,N)}-\ceil*{\log_2 N}.
\]
Computing the exact value of $\abs{E_1(b,N)}$ is quite involved, and instead we lower bound it. The set $E_1(b,N)$ trivially contains the unique all-zero error pattern, as well as $N$ single-error patterns. For double-error patterns, let us first choose $\bfi$ such that $b-1\leq i_\ell \leq n-b$. Now, any choice of $\bfj=\bfi+\bfeps$, with $0<\sum_{\ell\in[D]}\abs{\eps_\ell}\leq b-1$, is guaranteed to give us $\bfj\in [n]^D$. By~\cite[Theorem 4]{GolWel70}, the number of ways of choosing $\bfeps$ (and hence, choosing $\bfj$) is
\[
\sum_{\ell = 0}^{\min\set{D,b-1}} 2^{\ell}\binom{D}{\ell}\binom{b-1}{\ell}-1.
\]
Since each pair $\set{\bfi,\bfj}$ may be counted at most twice (depending on the order of choosing $\bfi$ and $\bfj$), we have
\begin{align}
\abs*{E_1(b,N)} &\geq 1+n^D+\frac{1}{2}(n-2b+2)^D\parenv*{\sum_{\ell=0}^{\min\set{D,b-1}}2^\ell \binom{D}{\ell}\binom{b-1}{\ell}-1} \nonumber \\
&\geq \frac{1}{2}n^D\parenv*{1-\frac{2b-2}{n}}^D\cdot\sum_{\ell=0}^{\min\set{D,b-1}}2^\ell \binom{D}{\ell}\binom{b-1}{\ell} \nonumber \\
&\overset{(a)}{\geq} \frac{1}{2}n^D\parenv*{1-\frac{2D(b-1)}{n}}\cdot\sum_{\ell=0}^{\min\set{D,b-1}}2^\ell \binom{D}{\ell}\binom{b-1}{\ell} \nonumber \\
&\overset{(b)}{\geq} \frac{1}{4}n^D \cdot \begin{cases}
2^{b-1} \binom{D}{b-1} & D \geq b-1, \\
2^D \binom{b-1}{D} & D < b-1,
\end{cases}
\label{eq:improve}
\end{align}
where $(a)$ follows from Bernoulli's inequality, and $(b)$ uses $n\geq 4D(b-1)$.

We can now use this in the expression for $\xi(C)$ and obtain
\begin{align*}
\xi(C) &\geq
\begin{cases}
b-1 + \log_2 \binom{D}{b-1} - 3 & D\geq b-1, \\
D + \log_2 \binom{b-1}{D} - 3 & D<b-1,
\end{cases} \\
&\geq
\begin{cases}
(b-1)(1+\log_2(D-b+2)-\log_2(b-1))-3 & D\geq b-1, \\
D(1+\log_2(b-D)-\log_2 D)-3 & D< b-1,
\end{cases}
\end{align*}
by using $\ceil{x}\leq x+1$, as well as
\begin{align*}
\binom{D}{b-1} & \geq \frac{(D-b+2)^{b-1}}{(b-1)^{b-1}}, &
\binom{b-1}{D} & \geq \frac{(b-D)^D}{D^D}.
\end{align*}
\end{proof}

\begin{remark}
We observe that inequality $(b)$ in~\eqref{eq:improve} reduces the sum to a single element (the last in the summation). When $D$ and $b$ are large, a more refined approach may reduce the sum to a carefully chosen summand, resulting in a better bound. To see that, denote 
\begin{align*}
Z&\eqdef \max\set{D,b-1}, & z &\eqdef \min\set{D,b-1},\\
\zeta &\eqdef \frac{z}{Z}, & \eta & \eqdef 1+\zeta - \sqrt{1+\zeta^2}.
\end{align*}
Recall the definition of the binary entropy function,
\[
H(x) \eqdef -x\log_2 x - (1-x)\log_2 (1-x),
\]
and the fact that~\cite[Ch.~10.11, Lemma 7]{MacSlo78}
\[
\binom{n}{\lambda n} = 2^{nH(\lambda)(1+o(1))},
\]
for all $0<\lambda< 1$ such that $\lambda n $ is an integer. Here, $o(1)$ denotes a vanishing function as $n\to\infty$.

Now, in $(b)$ of~\eqref{eq:improve}, we replace that sum with the $\ceil{\eta Z}$-th summand. With that replacement, we get
\[
\abs*{E_1(b,N)}\geq \frac{1}{4}n^D 2^{Z(\eta+H(\eta)+\zeta H(\eta/\zeta))(1+o(1))},
\]
and thus
\[
\xi(C) \geq Z(\eta+H(\eta)+\zeta H(\eta/\zeta))(1+o(1)) - 3,
\]
where $o(1)$ denotes a vanishing function as $Z=\max\set{D,b-1}\to\infty$.
\end{remark}

%%%%%%%%%%%%%%%%%%%%%%%%%%%%%%%%%%%%%%%%%%%%%%%%%%%%%%%%%%%%%%%%
%%%%%%%%%%%%%%%%%%%%%%%%%%%%%%%%%%%%%%%%%%%%%%%%%%%%%%%%%%%%%%%%
%%%%%%%%%%%%%%%%%%%%%%%%%%%%%%%%%%%%%%%%%%%%%%%%%%%%%%%%%%%%%%%%
\section{The Straight Model}
\label{sec:straight}
%%%%%%%%%%%%%%%%%%%%%%%%%%%%%%%%%%%%%%%%%%%%%%%%%%%%%%%%%%%%%%%%
%%%%%%%%%%%%%%%%%%%%%%%%%%%%%%%%%%%%%%%%%%%%%%%%%%%%%%%%%%%%%%%%
%%%%%%%%%%%%%%%%%%%%%%%%%%%%%%%%%%%%%%%%%%%%%%%%%%%%%%%%%%%%%%%%

The final model we consider is the straight model. A $2$-weight-limited $b$-burst in this model contains either a single bit flip, or two positions with a bit flip whose coordinates differ in a single place (i.e., the two erroneous positions are placed along an axis-parallel line). We provide a construction and a lower bound on the excess redundancy of such codes. 

The construction we present is inspired by~\cite[Construction D]{YaaEtz10}. However, our construction is more elaborate, making use of a packing design to construct the parity-check matrix of the code.

A $(t,k,v)$-packing design\footnote{sometimes denoted as $t$-$(v,k,1)$ packing design} is a pair $(X,\cB)$, where $X$ is a finite set, $\abs{X}=v$, whose elements are called \emph{points}, and $\cB$ is a set of $k$-subsets of $X$, whose elements are called \emph{blocks}. Finally, every $t$-subset of points is contained in at most one block.

\begin{construction}
\label{con:straight}
Let $n\geq b\geq 2$ and $D\geq 1$ be integers. Let $(X,\cB)$ be a $(2,b,v)$-packing design, with $X=[v]$, and whose number of blocks is $\abs{\cB}= D$, $\cB=\set{S_0,S_1,\dots,S_{D-1}}$. Let us now denote the contents of the $i$-th block by $S_i=\set{s_{i,0},s_{i,1},\dots,s_{i,b-1}}$, for all $i\in[D]$.

Define
\begin{align*}
m & \eqdef \ceil*{\log_2 (n^D+1)}, & 
a &\eqdef \ceil*{\log_2 (v+1)}.
\end{align*}
Let $\alpha\in \F_{2^m}$ be primitive element. Assume $A$ is a $2a\times v$ parity-check matrix for a (shortened) binary double-error-correcting BCH code, and denote its $i$-th column by $\bfa^\T_i$.

Let $H$ be a matrix, whose columns are indexed by $[n]^D$, and whose $\bfi$-th column is
\begin{equation}
\label{eq:Hstr}
\bfh^\T_\bfi = \begin{pmatrix}
1 \\
\sum_{\ell\in[D]} \bfa^\T_{s_{\ell,(i_\ell \bmod b)}} \\
\bfone\cdot\floor{\bfi^\T / b} \bmod 2 \\
\alpha^{[\bfi]_n}
\end{pmatrix}.
\end{equation}
We construct the $D$-dimensional binary code $C$, whose parity-check matrix is $H$.
\end{construction}

\begin{theorem}
\label{th:str}
Let $C$ be the code with parity-check matrix $H$ from Construction~\ref{con:straight}. Then $C$ is an $[N=n^D,k]$ $D$-dimensional code of length $n^{\times D}$, capable of correcting a single $2$-weight-limited $b$-burst in the straight model.
\end{theorem}

\begin{proof}
Assume $\bfc\in C$ was transmitted, and $\bfy=\bfc+\bfe$ was received, where $\bfe\in E_{\str}(b,N)$. The receiver computes the syndrome:
\[
\bfs^\T=H\bfy^\T=H\bfe^\T=\begin{pmatrix}
s_0 \\
\bfs^{\T}_1 \\
s_2 \\
s_3
\end{pmatrix}
\]
where the components correspond to those of $H$ from~\eqref{eq:Hstr}, i.e., $s_0,s_2\in\F_2$, $\bfs_1\in\F_2^{2a}$, and $s_3\in\F_{2^m}$. We distinguish between cases depending on $\wt(\bfe)$.

\textbf{Case 1:} $\wt(\bfe)=0$. In this case $\bfe=\bfzero$ and $\bfs=\bfzero$ and the decoder returns $\bfy$.

\textbf{Case 2:} $\wt(\bfe)=1$. Let $\bfi\in[n]^D$ be the location of the error, i.e., the unique index such that $e_{\bfi}=1$. Thus, $\bfs=\bfh_{\bfi}$. The decoder distinguishes this case from Case 1 by the fact that $s_0=1$. To complete this case, it remains to show that the columns of $H$ are all distinct. 

By Proposition~\ref{prop:baseq},
\[ 0\leq [\bfi]_n \leq n^D-1 < 2^{m}-1=\ord(\alpha),\]
and so distinct $\bfi$ result in distinct $\alpha^{[\bfi]_n}$.

\textbf{Case 3:} $\wt(\bfe)=2$. In this case two errors occurred. We denote their positions as $\bfi,\bfj\in[n]^D$, $\bfi\neq\bfj$, and we have
\[
\bfs^\T = \begin{pmatrix}
s_0 \\ \bfs^\T_1 \\ s_2 \\ s_3
\end{pmatrix}
=\begin{pmatrix}
0 \\
\sum_{\ell\in[D]} \bfa^\T_{s_{\ell,(i_\ell \bmod b)}} + \sum_{\ell\in[D]} \bfa^\T_{s_{\ell,(j_\ell \bmod b)}} \\
(\bfone \cdot \floor{\bfi^\T / b}+\bfone \cdot \floor{\bfj^\T / b}) \bmod 2 \\
\alpha^{[\bfi]_n}+\alpha^{[\bfj]_n}
\end{pmatrix}.
\]
We observe that we therefore have $s_0=0$ but $s_3\neq 0$ (by our observations in Case 2). This enables the receiver to completely distinguish between Case 1, Case 2, and Case 3.

By the definition of the straight model, 
\[
\bfj = \bfi + \mu \bfeps_\ell,
\]
where $\bfeps_\ell$ is the $\ell$-th unit vector, and $-(b-1) \leq \mu \leq b-1$ is an integer, $\mu\neq 0$.  It then follows that we can write
\[
\bfs^\T_1 = \bfa^\T_{s_{\ell,(i_\ell \bmod b)}} +  \bfa^\T_{s_{\ell,(j_\ell \bmod b)}},
\]
namely, $s_1$ is the sum of two columns of $A$. But $A$ is a parity-check matrix for a (shortened) binary double-error-correcting BCH code, and so the identity of these columns may be recovered by the receiver, say $\bfa^\T_u$ and $\bfa^\T_w$, for $u,w\in[v]$. Notice that $u\neq w$ since $\abs{\mu}\leq b-1$, while $i_\ell$ and $j_\ell$ are taken modulo $b$ when computing $s_1$.

By the properties of the $(2,b,v)$-packing design, $u$ and $w$ are contained in at most one block, which the receiver can now identify as block $\ell$. Also, knowing $u$ and $w$ enables the receiver to uniquely recover $i_\ell \bmod b$ and $j_\ell \bmod b$. Since $j_\ell = i_\ell + \mu$, by using $s_2$ the receiver can now recover $\mu$ in the following way,
\begin{align*}
\mu'&\eqdef (j_\ell \bmod b)-(i_\ell \bmod b), \\
\mu &= \begin{cases}
\mu' & s_2=0, \\
\mu'+b & s_2=1, \mu'<0, \\
\mu'-b & s_2=1, \mu'>0.
\end{cases}
\end{align*}

Finally, using the last part of the syndrome, we can now write
\begin{align*}
s_3 &= \alpha^{[\bfi]_n}+\alpha^{[\bfj]_n} = 
\alpha^{[\bfi]_n} + \alpha^{[\bfi+\mu\bfeps_\ell]_n} = \alpha^{[\bfi]_n} \cdot\parenv*{1+\alpha^{[\mu\bfeps_\ell]_n}}
\end{align*}
Recall that $0<\abs{\mu}\leq b-1$, hence
\[
0 < \abs*{[\mu\bfeps_\ell]_n} \leq (b-1)n^{D-1} < 2^m-1 \leq \ord(\alpha),
\]
and so
\[
1+\alpha^{[\mu\bfeps_\ell]_n}\neq 0.
\]
The receiver now needs to solve
\[
\alpha^{[\bfi]_n}=\frac{s_3}{1+\alpha^{[\mu\bfeps_\ell]_n}}.
\]
By proposition~\ref{prop:baseq}, this has a unique solution, and the receiver knows the two positions in error, $\bfi$ and $\bfj$. It can then flip the bits in these two positions in $\bfy$ and correctly recover the transmitted codeword $\bfc$.
\end{proof}

Construction~\ref{con:straight} depends on the choice of packing design, with different designs resulting in different excess redundancy, and perhaps parameter restrictions. We give two examples for choosing the packing design. First, a trivial packing:

\begin{corollary}
\label{cor:strxi1}
For all $b\geq 2$ and $D\geq 1$, the code $C$ from Construction~\ref{con:straight} may be constructed to have excess redundancy
\[
\xi(C) \leq 2\log_2 b + 2 \log_2 D + 5. 
\]
\end{corollary}
\begin{proof}
We choose a trivial $(1,b,Db)$-packing design (which is obviously also $(2,b,Db)$), by partitioning $[Db]$ into $D$ arbitrary disjoint subsets to serve as blocks. Thus,
\begin{align*}
\xi(C) &= 2+2\ceil*{\log_2(Db+1)} + \ceil*{\log_2(n^D+1)} - \ceil*{\log_2 (n^D)} \\
& \leq 2 \log_2 b + 2\log_2 D + 5,
\end{align*}
by using $\ceil{\log_u(x+1)}\leq \log_u x + 1$ for positive integers $u,x$, as well as $x\leq \ceil{x}\leq x+1$ for all real $x$.
\end{proof}

Next, we use a Steiner system:

\begin{corollary}
\label{cor:strxi2}
For all $b\geq 2$ and $D\geq 1$, the code $C$ from Construction~\ref{con:straight} may be constructed to have excess redundancy
\[
\xi(C) \leq 4\log_2(b-1) + \log_2 D + 9.
\]
\end{corollary}
\begin{proof}
We make use of an optimal packing design as follows. For all prime power $q$, and any integer $s\geq 2$, there exists a $(2,q,q^s)$-packing design with $q^{s-1}(1+q+q^2+\dots+q^{s-1})$ blocks (this is in fact a Steiner system, see~\cite[Theorem 5.11]{ColDin07}).

To make things concrete, let $q$ be the smallest prime power such that $q\geq b$. Thus, by Bertrand's postulate,
\[ b \leq q \leq 2b-3 \leq 2(b-1).\]
Let $s\geq 2$ be the smallest integer such that $q^{2(s-1)} \geq D$, hence,
\[
q^{2(s-2)} < D \leq q^{2(s-1)}.
\]
Construct the $(2,q,q^s)$-packing design, and then arbitrarily keep only $D$ of its blocks, and in each such block, arbitrarily keep only $b$ elements. We note that the number of points in the design, $q^s$, satisfies
\[
q^s \leq \sqrt{D}\cdot q^2 \leq \sqrt{D} \cdot 4(b-1)^2.
\]
Using this in the definition for $\xi(C)$ we get
\begin{align*}
\xi(C) &= 2+2\ceil*{\log_2(q^s+1)} + \ceil*{\log_2(n^D+1)} - \ceil*{\log_2 (n^D)} \\
& \leq 4 \log_2 (b-1) + \log_2 D + 9,
\end{align*}
by using $\ceil{\log_u(x+1)}\leq \log_u x + 1$ for positive integers $u,x$, as well as $x\leq \ceil{x}\leq x+1$ for all real $x$.
\end{proof}

For a lower bound on the excess redundacy we state and prove the following theorem.

\begin{theorem}
\label{th:strxilower}
Let $D\geq 1$ and $n\geq b\geq 2$ be integers, and let $C$ be an $[N=n^D,k]$ $D$-dimensional code of length $n^{\times D}$, that is capable of correcting a single $2$-weight-limited $b$-burst in the straight model. Then its excess redundancy is lower bounded by
\[
\xi(C) \geq \log_2 (b-1) + \log_2 D -2.
\]
\end{theorem}

\begin{proof}
Again, we follow~\eqref{eq:rlowbound}, and write
\[
\xi(C) \geq \log_2\abs*{E_{\str}(b,N)}-\ceil*{\log_2 N}.
\]
The error patterns in $E_{\str}(b,N)$ contain the unique all-zero pattern, $N=n^D$ single-error patterns, and all two-error patterns, where the positions differ in a single coordinate by at most $b-1$. Thus,
\begin{align}
\abs*{E_\str(b,N)} &=1+n^D + \sum_{i=1}^{b-1} D n^{D-1}(n-i) = 1 + ((b-1)D+1)n^D - Dn^{D-1}\frac{b(b-1)}{2} \nonumber \\
& \geq \frac{b-1}{2}D n^D, \label{eq:strlow}
\end{align}
where for the last inequality we used $n\geq b$. Putting everything together we get
\[
\xi(C) \geq \log_2\abs*{E_{\str}(b,N)}-\ceil*{\log_2 n^D}
\geq \log_2 (b-1) + \log_2 D -2.
\]
\end{proof}

\begin{remark}
\label{rem:strbound}
In Theorem~\ref{th:strxilower}, if we require $n\geq D(b^2-b)/2$, then~\eqref{eq:strlow} may be further improved to $\abs{E_\str(b,N)}\geq (b-1)Dn^D$, resulting in a slight improvement,
\[
\xi(C) \geq \log_2 (b-1) + \log_2 D -1.
\]
\end{remark}

%%%%%%%%%%%%%%%%%%%%%%%%%%%%%%%%%%%%%%%%%%%%%%%%%%%%%%%%%%%%%%%%
%%%%%%%%%%%%%%%%%%%%%%%%%%%%%%%%%%%%%%%%%%%%%%%%%%%%%%%%%%%%%%%%
%%%%%%%%%%%%%%%%%%%%%%%%%%%%%%%%%%%%%%%%%%%%%%%%%%%%%%%%%%%%%%%%
\section{Conclusion}
\label{sec:conc}
%%%%%%%%%%%%%%%%%%%%%%%%%%%%%%%%%%%%%%%%%%%%%%%%%%%%%%%%%%%%%%%%
%%%%%%%%%%%%%%%%%%%%%%%%%%%%%%%%%%%%%%%%%%%%%%%%%%%%%%%%%%%%%%%%
%%%%%%%%%%%%%%%%%%%%%%%%%%%%%%%%%%%%%%%%%%%%%%%%%%%%%%%%%%%%%%%%

In this paper we studied codes capable of correcting a single $2$-weight-limited burst of size $b$ in different models. We constructed codes and computed their excess redundancy. We also proved lower bounds on the excess redundancy in the various models. The results are summarized in Table~\ref{tab:summary}.

\begin{table}[t]
\caption{A summary of the excess redundancy achieved by constructions for $D$-dimensional codes capable of correcting a single $2$-weight-limited $b$-burst, and general lower bounds on it. Some bounds require the code length to be large enough.}
\label{tab:summary}
\begin{tabular}{cll}
\toprule%
Model & $\xi(C)$ & Location  \\
%%%%%%%%%%%%%%%%%%%%%%%%%%%%%%%%%%%%%%%%%%%%%
\midrule
%%%%%%%%%%%%%%%%%%%%%%%%%%%%%%%%%%%%%%%%%%%%%
\multirow{5}{*}{$L_\infty$} &
=$\ceil{\log_2(b+1)}$ &
\cite{Yaa07},\cite[Th.~6]{EtzYaa09}, only $D=1$ \footnotemark{} \\
\cmidrule{2-3}
&
$\leq 3\ceil{2\log_2 b}+3$ &
\cite[Th.~7]{EtzYaa09}, only $D=2$ \\
\cmidrule{2-3}
&
$\leq
2\ceil{\log_2(b^D+1)}+D+1$ &
Construction~\ref{con:linf}, Corollary~\ref{cor:linfxi} \\
\cmidrule{2-3}
&
$\leq D\log_2(2b)+3$ &
Construction~\ref{con:linfext}, Corollary~\ref{cor:linfxiext} \footnotemark{} \\
\cmidrule{2-3}
&
$\geq D\log_2(2b-1)-2$ &
Theorem~\ref{th:linfxilower} \\
%%%%%%%%%%%%%%%%%%%%%%%%%%%%%%%%%%%%%%%%%%%%%
\midrule
%%%%%%%%%%%%%%%%%%%%%%%%%%%%%%%%%%%%%%%%%%%%%
\multirow{4}{*}{$L_1$} &
$\leq \ceil{\log_2 D}+1$
&
\cite[Th.~4]{SchEtz05}, only $b=2$ \\
\cmidrule{2-3}
&
$\leq 2b \log_2 b + 2(b-1)\log_2(D+1) + 4b+D+4$
&
Construction~\ref{con:l1}, Corollary~\ref{cor:l1xi} \footnotemark{} \\
\cmidrule{2-3}
&
$\leq 3\log_2 D + 6$
&
Remark~\ref{rem:l1b3}, only for $b=3$ \\
\cmidrule{2-3}
&
$\geq
\begin{cases}
b-1 + \log_2 \binom{D}{b-1} - 3 & D\geq b-1, \\
D + \log_2 \binom{b-1}{D} - 3 & D<b-1,
\end{cases}$ &
Theorem~\ref{th:l1xilower} \\
%%%%%%%%%%%%%%%%%%%%%%%%%%%%%%%%%%%%%%%%%%%%%
\midrule
%%%%%%%%%%%%%%%%%%%%%%%%%%%%%%%%%%%%%%%%%%%%%
\multirow{4}{*}{straight} &
$\leq \ceil{\log_2 D}+1$
&
\cite[Th.~4]{SchEtz05}, only $b=2$ \\
\cmidrule{2-3}
&
$\leq 2\log_2 b + 2\log_2 D + 5$
&
Construction~\ref{con:straight}, Corollary~\ref{cor:strxi1} \\
\cmidrule{2-3}
&
$\leq 4\log_2 (b-1) + \log_2 D + 9$
&
Construction~\ref{con:straight}, Corollary~\ref{cor:strxi2} \\
\cmidrule{2-3}
&
$\geq \log_2 (b-1) + \log_2 D -2$ &
Theorem~\ref{th:strxilower} \\
\botrule
\end{tabular}
\end{table}

\addtocounter{footnote}{-2}
\footnotetext{In the one-dimensional case, $D=1$, all models are equal, and the line in the table pertains to all.}
\addtocounter{footnote}{1}
\footnotetext{The codewords are of size $(nb)^{\times D}$, and it is required that $\gcd(b,2^m-1)=1$,  where $m$ is from~\eqref{eq:linfm}.}
\addtocounter{footnote}{1}
\footnotetext{The codewords are of size $(np)^{\times D}$, and it is required that $\gcd(p,2^m-1)=1$, where $p$ is the smallest prime $\geq 2b+1$, and $m$ is from~\eqref{eq:l1m}.}

We would like to point out a few conclusions from the results. In the $L_\infty$ model, the best construction has excess redundancy almost matching the lower bound. Construction~\ref{con:linfext} obtains a lower excess redundancy than Construction~\ref{con:linf} by limiting the possible values of $n$. A similar approach is seen in Construction~\ref{con:l1} for the $L_1$ model. We comment that the restriction on $n$ in Construction~\ref{con:l1} may be relaxed at the cost of increased excess redundancy.

In contrast, it the largest gap between construction and bound seems to occur in the $L_1$ model. We note that the error models are nested in the sense that the error patterns allowable in the straight model are contained in those of the $L_1$ model, which in turn, are contained in those of the $L_\infty$ model. Thus, we can use the constructions in the $L_\infty$ model to correct errors in the $L_1$ model. In some cases this might even result in a lower excess redundancy. By closely checking the excess redundancy, given any $b$, for all sufficiently large $D$ (depending on $b$), the upper bound on the excess redundancy of Construction~\ref{con:l1} for the $L_1$ model (see Corollary~\ref{cor:l1xi}) is smaller than the lower bound on the excess redundancy in the $L_\infty$ model (see Theorem~\ref{th:linfxilower}). Thus, despite the gap it seems that Construction~\ref{con:l1} has merit.

Finally, the straight-model results depend on the choice of a packing design. The designs used in Corollary~\ref{cor:strxi1} and Corollary~\ref{cor:strxi2} are two extremes: a trivial set of disjoint blocks, versus a Steiner system. More generally, a $(2,k,v)$-packing design is equivalent to a binary constant-weight code of length $v$, weight $k$, and minimum Hamming distance $2(k-1)$. However, constant-weight codes such as~\cite{GraSlo80} seem to provide inferior results.

\backmatter

\bmhead*{Acknowledgments}
This work was supported in part by the Israel Science Foundation (Grant No. 1789/23).

\bibliography{allbib}

\end{document}